\documentclass{article}

\usepackage[english]{babel}

\usepackage[a4paper, top=2cm,bottom=2cm,left=3cm,right=3cm,marginparwidth=1.75cm]{geometry}

\usepackage{amsmath}
\usepackage{graphicx}
\usepackage[colorlinks=true, allcolors=blue]{hyperref}
\usepackage[utf8]{inputenc}
\usepackage{amsthm,amsmath,latexsym,amssymb,amsfonts,amscd}
\usepackage{authblk}

\usepackage[all]{xy}
\newtheorem{theorem}{Theorem}[section]

\newtheorem{proposition}[theorem]{Proposition}

\theoremstyle{definition}

\newtheorem{example}[theorem]{Example}

\newtheorem{remark}[theorem]{Remark}

\newcommand{\PP}{{\mathcal{P}}}
\newcommand{\MM}{{\mathrm{Mat}}}
\newcommand{\HH}{{\mathrm{Hom}}}
\newcommand{\FF}{{\mathcal{F}}}
\newcommand{\cc}{{\mathrm{coker}}}
\newcommand{\OO}{{\mathrm{Op}}}
\newcommand{\RR}{{\mathcal{R}}}

\title{\bf On auxiliary fields and Lagrangians \\for  relativistic wave equations}

\author[1,2]{{Alexey \textsc{Sharapov}}\thanks{E-mail: sharapov@phys.tsu.ru}}
\author[2]{ David \textsc{Shcherbatov}}
\affil[1] {\small \it Centro de Matem\'atica, Computa\c{c}\~{a}o e
Cogni\c{c}\~{a}o, Universidade Federal do ABC, Santo Andr\'e, SP, 
Brazil} 
\affil[2]{\it Physics Faculty, Tomsk State University, Lenin ave. 36, Tomsk 634050, Russia}

\date{}
\begin{document}
\maketitle

\begin{abstract}
We address the problem of the existence of a Lagrangian for a given system of linear PDEs with constant coefficients. As a subtask, this involves bringing the system into a pre-Lagrangian form, wherein the number of equations matches the number of unknowns. We introduce a class of overdetermined systems, called co-flat, and show that they always admit a pre-Lagrangian form, which can be explicitly constructed by means of auxiliary variables. Moreover, we argue that such systems enjoy pre-Lagrangian formulations without auxiliary variables at all. As an application of our method, we construct new pre-Lagrangian and Lagrangian formulations for free massive fields of arbitrary integer spin. In contrast to the well-known models of Singh and Hagen, our Lagrangians involve much fewer auxiliary fields. 

\end{abstract}

\section{Introduction}

Most of the equations encountered in physics, especially in the theory of fundamental interactions, come from the least action principle. This is not surprising, since the Lagrangian formulation of classical dynamics is generally considered as the first step towards quantum theory. Therefore, the following problem often occurs: given a system of partial differential equations (PDEs), find a Lagrangian it comes from. This is known as the inverse problem of the calculus of variations. 

A typical example of the above situation is provided by the relativistic wave equations for higher spin particles. Although at the free level, such equations can be systematically derived and classified using the representation theory of the Poincar\'e group, the Lagrangian formulation for them  had long remained an open problem. The point is that, unlike the classical inverse problem studied in mathematics, the number of field equations usually exceeds the number of fields. In other words, the PDE systems of physical interest are usually overdetermined and finding their Lagrangian formulations is not reduced to the problem of `variational multiplier’. 
For instance, a free massive particle of spin $s$ can be described by a traceless symmetric tensor field $\phi_{\mu_1\cdots\mu_s}$ subject to the equations of motion
\begin{equation}\label{KG}
    (\square -m^2)\phi_{\mu_1\cdots\mu_s}=0\,,\qquad \partial^{\mu_1}\phi_{\mu_1\cdots \mu_s}=0\,.
\end{equation}
Here $\square=\partial^\mu\partial_\mu$ is the D'Alembert operator and $m\neq 0$ is the particle's mass. In $4$-dimensional Minkowski space, this yields $(s+1)^2+s^2$ 
equations for $(s+1)^2$ field components.

As was first noted by Fierz and Pauli \cite{Fierz:1939ix}, one can correct the imbalance between the fields and equations by means of the so-called {\it auxiliary fields}. The idea is to introduce a suitable number of additional fields so that they all vanish due to the equations of motion, whereupon the field equations become equivalent to Eqs. (\ref{KG}). In particular, a minimal Lagrangian formulation for the massive spin-2 field was shown to involve a single  auxiliary scalar. 
Later, developing the ideas of Fierz and Pauli, Singh and Hagen proposed Lagrangians for free massive fields of arbitrary spin in four dimensions \cite{Singh:1974qz, Singh:1974rc}. As auxiliary fields for system (\ref{KG}), they used symmetric and traceless tensors of ranks $0$, $1$, ..., $s-2$ (one of each). The massless limit of the Singh--Hagen Lagrangians for integer spins was considered by Fronsdal \cite{Fronsdal:1978rb}. In the massless limit, the Lagrangian equations become underdetermined (enjoy gauge invariance) and all but one of the auxiliary tensor fields decouple. Thus, both massive and massless particles of higher spins need auxiliary fields for their Lagrangian formulation.

In this paper, we address the inverse problem of the calculus of variations for overdetermined systems of PDEs with an eye to applications in relativistic field theory. Since the inverse problem for such systems is still in its infancy, we restrict ourselves to linear PDE systems with constant coefficients. 
Even in such a restricted formulation, the problem is far from being trivial and its complete solution is elusive. In general, one can divide the problem into two parts. At the first stage, it is necessary to extend the original PDE system with a suitable set of auxiliary fields. The extended system must satisfy two conditions: (i) be equivalent to the original system and (ii) be square in size, so that the total number of fields is equal to the number of field equations.
We call such a system {\it pre-Lagrangian}. Generally, the matrix differential operator determining the extended system, being square, may not be formally self-adjoint, in which case the  pre-Lagrangian equations are not Lagrangian {\it per se}.  Therefore, at the next stage, one needs to bring the pre-Lagrangian system into a Lagrangian one by means of a variational multiplier. Unlike much of the mathematical literature (see e.g. \cite{1984JDE, H, biesecker2009inverse,  S}), we look for multipliers that are square matrices with entries in differential operators, rather than functions. The variational multiplier should be chosen in such a way as to make the operator of the pre-Lagrangian system formally self-adjoint. Furthermore, the matrix of the multiplier must be unimodular to ensure equivalence. The last two conditions are not easy to satisfy, and we are unaware of any  practical criterion for deciding whether a given operator admits a variational multiplier. Some necessary conditions are discussed below in Sec. 2.3. This drawback motivates us to focus on the first step, i.e., the existence of a pre-Lagrangian form for a given system of linear PDEs with constant coefficients. In this formulation, the problem can be systematically studied by methods of homological algebra. All necessary mathematics is reviewed in the next section.

Using the language  of homological algebra, we introduce the notion of a {\it co-flat system}. Loosely speaking, the co-flat systems represent the simplest class of overdetermined PDE systems. For such systems, a systematic method is proposed for introducing auxiliary fields that make them pre-Lagrangian. Despite their simplicity, the co-flat systems cover a wide class of physically interesting theories. For example, all massive higher spin fields (\ref{KG}) fall into this class. In that case, the auxiliary fields for spin-$s$ particles, as offered by our method, constitute a symmetric traceless tensor of rank $s-1$. This is much less than in the Singh--Hagen models. The price to be paid for such simplification is that the pre-Lagrangian equations necessarily include higher derivatives. 
Direct calculations in particular cases show that these pre-Lagrangian equations are equivalent to Lagrangian ones so that the variational multiplier problem has a positive solution for them. This gives new Lagrangians for higher spin massive theories with fewer auxiliary fields. 

It might be tempting to assume that a single Lorentz irreducible tensor $\psi_{\mu_1\cdots\mu_{s-1}}$ provides the minimum number of auxiliary fields required for spin $s>2$. Surprisingly, this is not true. It follows from the Quillen--Suslin theorem that any co-flat system  admits a pre-Lagrangian formulation without auxiliary fields at all! 
Unfortunately, the theorem does not guarantee that the corresponding equations will be Lorentz invariant (although their solutions carry a unitary irreducible representation of the Poincar\'e group). An examination of simple examples shows that this is indeed the case.
Perhaps this points to some mechanism of spontaneous  breaking of Lorentz symmetry in nonlinear models with higher-spin fields. 
In any case, the very existence of relativistic wave equations with broken Lorentz symmetry looks intriguing and deserves further study.

\section{Algebraic preliminaries}
In this section, we fix our notation and briefly recall some algebraic notions and constructions that will be needed later. For a systematic exposition of the subject, we refer the reader to \cite{book:Lam2, Lam1, Palamod, Tarkh}. 

\subsection{Free resolutions}
Let  $\PP=\mathbb{R}[p_1,\ldots,p_d]$ be the ring of polynomials in $d\geq 1$  indeterminates over  reals.  By $\MM_{n,m}(\PP)$ we denote the set of all $n\times m$-matrices with entries in $\PP$. In this paper, we are mostly interested in the systems of homogeneous linear equations over the polynomial ring $\PP$. Each such system is determined by a matrix $A\in \MM_{n,m}(\PP)$ and has the form $AX=0$ for unknown $X\in \MM_{m,1}(\PP)$.  More abstractly, one can think of the matrix $A$ as defining a homomorphism $A: \PP^m\rightarrow \PP^n$  of free $\PP$-modules\footnote{We recall that a module is an algebraic structure defined by the same axioms as for a vector space but where the scalars belong to a ring (in our case, $\PP$) rather than a  field.} of ranks $m$ and $n$. In the following, we will frequently identify the $\PP$-module $\HH_\PP(\PP^m, \PP^n)$ with $\MM_{n,m}(\PP)$.  Associated with the homomorphism $A$ is the pair  of $\PP$-modules $\ker A\subset \PP^m$ and $\cc A=\PP^n/\mathrm{Im}A$. In general, neither module is free. The former is identified with the module of solutions to the  above system of  linear equations. The latter module admits a finite free resolution 
\begin{equation}\label{ResA}
    \xymatrix{ 0\ar[r]&
     \PP^{m_{N}}\ar[r]^{A_{N-1}}&\cdots \ar[r]&\PP^{m_2}\ar[r]^{A_1}& \PP^{m_1}\ar[r]^-{A_0}&\PP^{m_0}\ar[r]^-\pi&\mathrm{coker} A\ar[r]&0\,,
    }
\end{equation}
where $\ker A_{i-1}=\mathrm{Im}\, A_{i}$ (exactness) and $N\leq d$. This is the content of {\it Hilbert's syzygy theorem}. The matrix $A_0$ defines a {\it finite presentation} of the module $\mathrm{coker} A$. In particular, one can always take $m_1=m$, $m_0=n$, and $A_0=A$. Then the exactness of the sequence (\ref{ResA}) implies that any solution to the equation $AX=0$ can be written as $X=A_1Y$ for some $Y\in \MM_{m_2,1}(\PP)$. Thus, the sequence of matrices $(A_1, A_2, \ldots, A_{N-1})$ may be thought of as parameterizing the `general
solution' to the linear system $AX=0$. By abuse of  language, we will sometimes refer to  the exact sequence (\ref{ResA}) as a (left)  resolution of  $A$.  The number $N$ is called the {\it length} of the resolution. Neither the resolution (\ref{ResA}) nor its length is uniquely determined by the matrix $A$. The only invariant meaning has the minimum length of a free resolution. It is called the {\it projective dimension} of the module $\cc\, A$. The projective dimension measures how far a $\PP$-module  is from being free.

We say that the two matrices $A\in \MM_{n,m}(\PP)$ and $B\in \MM_{s,l}(\PP)$ are {\it equivalent} (in writing $A\sim B$) if $\cc\, A\simeq \cc\, B$ as $\PP$-modules. It is known that the last condition  takes place iff the corresponding free resolutions are homotopy equivalent, that is, there exist homomorphisms $F_i$, $G_i$, $h^A_i$, and $h^B_i$ fitting the diagram
$$
    \xymatrixcolsep{5pc}\xymatrix{
    \ldots\ar@/^/[r]^{A_2}& \PP^{m_2}\ar@/^/[r]^{A_1}\ar@/^/[l]^{h^A_{2}}\ar@/^/[d]^{G_{2}}& \PP^m\ar@/^/[r]^{A}\ar@/^/[l]^{h^A_{1}}\ar@/^/[d]^{G_1}& \PP^n\ar[r]^\pi\ar@/^/[l]^{h^A_0}\ar@/^/[d]^{G_{0}}& \cc\, A\ar@/^/[d]^{g}\\
    \ldots\ar@/^/[r]^{B_2}& \PP^{l_2}\ar@/^/[r]^{B_1}\ar@/^/[l]^{h^B_{2}}\ar@/^/[u]^{F_{2}}& \PP^l\ar@/^/[r]^{B}\ar@/^/[l]^{h^B_{1}}\ar@/^/[u]^{F_1}& \PP^s\ar[r]^\pi\ar@/^/[l]^{h^B_{0}}\ar@/^/[u]^{F_{0}}& {\cc\, B}\ar@/^/[u]^{f}    }    
$$
such that
\begin{equation}\label{FGh}
\begin{array}{ll}
    A_iF_{i+1}=F_{i}B_i\,,&  B_iG_{i+1}= G_{i}A_i\,,\\[3mm]
    h^A_iA_i +A_{i+1}h^A_{i+1} =1-F_{i+1}G_{i+1}\,,\quad& B_{i+1}h^B_{i+1}+h^B_iB_i=1-G_{i+1}F_{i+1}
    \end{array}
\end{equation}
for $i=0,1,2,\ldots$ and $A_0=A$. If one  regards Rels. (\ref{FGh}) as equations for unknown homomorphisms $F_i$, $G_i$, $h^A_i$, and $h^B_i$,  then the existence of a solution to the equations with $i=0$ ensures the existence of all other homomorphisms, see e.g. \cite[Prop. 1.2.8]{Tarkh}.  The mutually inverse isomorphisms $$f:\cc\, B\rightarrow \cc\, A\,,\qquad g:\cc\,A\rightarrow \cc\,B$$ are induced  by the homomorphisms $F_{0}$ and $G_{0}$, respectively.

When doing linear algebra over  polynomial rings, it might be useful to extend the ring $\PP$ to its field of fractions $\RR=\mathbb{R}(p_1,\ldots,p_d)$. The elements of the field $\RR$ are just rational functions in $p$'s. Since $\PP\subset \RR$, the field $\RR$ comes with the natural structure of a $\PP$-module. This module is known to be {\it flat}, meaning that the covariant functor $\RR\otimes_\PP -$ maps exact sequences of $\PP$-modules to exact.  The obvious isomorphism $\RR\otimes_\PP \PP^n\simeq \RR^n$ allows  one to pass easily from the free $\PP$-module $\PP^n$ to the $\RR$-vector space $\RR^n$. 

Recall that a $\PP$-module $\FF$ is called {\it injective} if the  contravariant functor $\HH_\PP(-, \FF)$ transforms exact sequences to exact ones. Applying such a functor to the exact sequence (\ref{ResA}) with $A_0=A$ and using the obvious isomorphism $\HH_\PP(\PP^n,\FF)\simeq\FF^n$, one obtains the new exact sequence

\begin{equation}\label{RAF}
    \xymatrix{ 0&
    \FF^{m_N}\ar[l]& \cdots\ar[l]_-{A^t_{N}}& \FF^m\ar[l]_-{A^t_1}&\FF^n \ar[l]_-{A^t}&\HH_\PP(\cc\, A, \FF)\ar[l]&0\ar[l]\,,
    }
\end{equation}
where by $A^t$ we denote the matrix transpose to $A$. In particular,
\begin{equation}\label{Mal}
    \ker_\FF A^t\simeq \HH_\PP(\cc\, A, \FF)\,.
    \end{equation}
This isomorphism is due to Malgrange \cite{Malgr}. It establishes a relation between the solutions to the linear system $A^tX=0$ for  $X\in \MM_{n,1}(\FF)$ and the module $\cc\, A$. 

An injective $\PP$-module $\FF$ is called an {\it injective cogenerator} if $\HH_\PP(M, \FF)\neq 0$ for any nontrivial $\PP$-module $M$. A typical example of an injective cogenerator for the polynomial ring $\PP$ is provided by the space of formal power series $\PP^\ast=\mathbb{R}[[x_1,\ldots, x_d]]$, the dual of the $k$-vector space $\PP$. Another example is the smooth real-valued functions $C^\infty(\Omega)$ on an open convex domain $\Omega\subset \mathbb{R}^d$. In both  cases the generators of the ring act as partial derivatives, $p_i f={\partial f}/{\partial x^i}$. It follows from the definition that once $\cc\, A\neq 0$, the equation  $A^t X=0$ has a nonzero solution whenever one looks for an $n\times 1$-matrix  $X$ with entries in an injective cogenerator. Roughly speaking, the injective cogenerator condition on $\cal F$ plays the same role as the condition of an algebraically closed base field in classical algebraic geometry.

\subsection{Quillen--Suslin theorem}
A matrix $A\in \MM_{n,m}(\PP)$ is called {\it unimodular} if it admits either a left or  right inverse matrix $B$, that is,
\begin{equation}
    BA=1_m\quad \mbox{or}\quad AB=1_n\,.
\end{equation}
It is clear that the first (resp. second) identity implies $n\geq m$ (resp. $m\geq n$). We will denote the set of all $n\times m$ unimodular matrices by $U_{n,m}(\PP)$.  The unimodular  matrices of $U_{n,n}(\PP)$, being invertible, form a group denote by $U_n(\PP)$. The groups $U_n(\PP)$ and $U_m(\PP)$ naturally act on  $\MM_{n,m}(\PP)$ from the left and right, and one may wonder what is the  simplest form that every matrix of $\MM_{n,m}(\PP)$ can be reduced to by this action. 
In the case of unimodular matrices, the answer to this question is given by the Quillen--Suslin theorem \cite{Lam1}.

\begin{theorem}
    \label{theor6.1} For any $A\in U_{n, m}(\PP)$ with $n\leq m$, there exists a matrix $U\in U_m(\PP)$ such that 
    \begin{equation}\label{AU}
    AU=\left(\begin{array}{ccccccccc}
       1&0&\cdots&0&  0&\cdots &0\\
       0&1&\cdots&0&  0&\cdots &0\\
  \vdots&\vdots&\ddots&\vdots&\vdots&\ddots&\vdots\\
        0 &0& \cdots &1&0&\cdots&0
    \end{array}\right)=(1_n,  0)\,.
    \end{equation}
\end{theorem}
 The matrix on the right represents the canonical form of a unimodular matrix.  
 It follows from the definition of $U$ that
$$
U^{-1}=\left(\begin{array}{c}
    A \\
    A^\vee 
\end{array}\right)
$$ 
for some unimodular $(m-n)\times m$-matrix $A^\vee $. We will refer to $A^\vee$  as a {\it unimodular completion} of the matrix $A$. This gives an alternative formulation of the Quillen--Suslin theorem: Every unimodular matrix can be completed to an invertible one. Writing the matrix $U$ in the block form $U=(B, B^\vee)$, where $B^\vee$ consists of the last $m-n$ columns of $U$, one finds that $AB=1_n$ and $AB^\vee=0$.  The columns of the unimodular matrix $B^\vee$ define the basis of $\ker A$ and we arrive at the short exact sequence  
\begin{equation}\label{SRA}
   \xymatrix{ 0\ar[r] & \PP^{m-n} \ar[r]^{B^\vee}& \PP^m \ar[r]^{A}& \PP^n\ar[r] &0}\,,
\end{equation}
which can be viewed as the Hilbert resolution of the trivial module $\cc\, A=0$. By construction, the sequence splits as $\PP^m=\mathrm{Im}\, B\oplus \mathrm{Im}\, B^\vee$. 

Finally, a unimodular matrix $A\in U_n(\PP)$ is called {\it elementary} if it has either the upper or lower triangular form. For $n\neq 2$, the elementary matrices are known to generate the whole group $U_n(\PP)$, see \cite[Ch. VI.4.5]{Lam1}. This means that one can bring every unimodular matrix into the canonical form by a sequence of elementary column (or row) transformations.  The case $n=2$ is special. For example, it is known that the  following unimodular matrix 
\begin{equation}
    \left(\begin{array}{cc}
    1+xy&x^2\\
    -y^2& 1-xy
    \end{array}\right)
\end{equation}
cannot be written as a finite product of elementary matrices \cite[Ch. I.8]{Lam1}.

\subsection{Resolutions of differential operators with constant coefficients}
Let  $\mathcal{F}$ denote the ring of smooth  functions on an open convex domain $\Omega\subset \mathbb{R}^d$. Then the free $\mathcal{F}$-module $\mathcal{F}^n$ consists of collections of functions $f=(f_1,\ldots, f_n)$ with $f_a\in C^\infty(\Omega)$.  Denote by $\partial_i=\partial/\partial x^i$ the partial derivatives w.r.t. the Cartesian coordinates $(x^1,\ldots,x^d)$ in $\mathbb{R}^d$.
Given a matrix $A=\big(A_{a\alpha}(p)\big)$ of $\MM_{n,m}(\PP)$, define the differential operator $\hat{A}: \FF^n\rightarrow \FF^m$ by the rule
$$
(\hat A f)_\alpha= \sum_{a=1}^n A_{a\alpha}(\partial) f_a\,,\qquad \alpha=1,\ldots,m\,,
$$
where $A_{a\alpha}(\partial)$ is obtained from $A_{a\alpha}(p)$ by the substitution $p_i\rightarrow \partial_i$.
We will refer to the matrix $A$ as the {\it symbol} of a differential operator $\hat A$. By differential operators on $\Omega$ with constant coefficients we understand the operators that are obtained by the above assignment $A\mapsto \hat A$.  The $\mathbb{R}$-space of all differential operators  $\hat A: \FF^n\rightarrow \FF^m$ with constant coefficients will be denoted by $\OO_{n,m}(\FF)$. If $\hat B$ is an operator of $\OO_{m,l}(\mathcal{F})$, then it follows from the definition that 
\begin{equation}
    \hat B \hat A=\widehat{AB}\,.
\end{equation}
(Notice the change of order.) 

The space of smooth functions $\FF$, considered as a $\PP$-module under the assignment $p_i\rightarrow \partial_i$, is known to be injective and the Malgrange isomorphism (\ref{Mal}) reduces the study of linear PDEs $\hat A f=0$ to the study of their symbols $A$. In particular, two equations have isomorphic solution spaces iff their symbols  are equivalent.  

In this paper, we focus on particular classes of operators. Of most interest to us are operators of $\OO_{n,n}(\FF)$, $n=1,2,\ldots$, which  we call {\it pre-Lagrangian}. By definition, the symbol of a pre-Lagrangian operator is a square matrix $A\in \MM_{n,n}(\PP)$. The commutative ring $\PP$ enjoys the involution $a(p)\mapsto a^\dagger(p)=a(-p)$, which extends to the involution of the $\PP$-module $\MM_{n,n}(\PP)$ as 
$A(p)\mapsto A^\dagger(p)=A^t(-p)$. If one regards $\MM_{n,n}(\PP)$ as a ring, then the involution above defines an anti-isomorphism of order two, i.e., $(AB)^\dagger=B^\dagger A^\dagger$. 

A pre-Lagrangian operator $\hat A\in \OO_{n,n}(\FF)$ is called {\it Lagrangian} if $A^\dagger=A$. 
The terminology is obvious: If $\hat A$ is Lagrangian, then, under suitable boundary conditions, one can obtain the equation $\hat Af=g$ for unknown $f\in \FF^n$ from the variational principle for the action functional 
$$
S[f]=\int_\Omega d^dx  \Big(\frac12 f^t\hat Af-f^tg\Big)\,.
$$
Notice that whenever (\ref{ResA}) defines a resolution of the symbol $A$, the nonhomogeneous equation $\hat Af=g$
implies the equality $\hat A_1 g=0$, called usually the compatibility condition.  Correspondingly, the operator $\hat A_1$ is referred to as a {\it compatibility operator}.  

\begin{example} Let us add sources to equations (\ref{KG}) to get the nonhomogeneous linear system  
    \begin{equation}\label{jj}
    (\square -m^2)\phi_{\mu_1\cdots\mu_s}=j'_{\mu_1\cdots\mu_s}\,,\qquad \partial^{\mu_1}\phi_{\mu_1\cdots \mu_s}=j''_{\mu_2\cdots \mu_s}\,.
\end{equation}
Applying the Klein--Gordon operator to the second equation and subtracting the divergence of the first one yields the following constraint on the sources: 
\begin{equation}
      (\square -m^2)j''_{\mu_2\cdots \mu_s}-\partial^{\mu_1}j'_{\mu_1\cdots\mu_s}  =0\,.
\end{equation}
The matrix differential operator determining this constraint is the compatibility operator for the wave operator in (\ref{jj}). 

\end{example}

We say that a pre-Lagrangian operator $\hat A$ is equivalent to a Lagrangian one if there exists a unimodular matrix $U$ such that the operator $\hat U\hat A$ is Lagrangian, i.e.,
\begin{equation}\label{ALA}
    AU=U^\dagger A^\dagger \,.
\end{equation}
Clearly, $ A U \sim A$, so that the equations $\hat U \hat A f=0$ and $\hat Af=0$ have the same solution space for any $\FF$. In the context of the inverse problem of the calculus of variations, the operator $\hat{U}$, making pre-Lagrangian differential equations Lagrangian, is called a {\it variational multiplier}. We do not know any effective criterion for deciding whether a given pre-Lagrangian operator  admits a variational multiplier. Computing the determinant of the left- and right-hand sides of (\ref{ALA}), one obtains the following necessary condition:
\begin{equation}
    \det A^\dagger=\det A\,.
\end{equation}
For example, the operator determining the ordinary differential equation $d^k f/dx^k=0$ is (equivalent to) Lagrangian iff $k$ is even. 
To formulate finer necessary conditions, we need more notation. Let $A\in \MM_{n,n}(\PP)$ and $r\leq n$. 
An $r$-minor of $A$ is the determinant of any $r\times r$-submatrix of $A$. Let $\Delta_r(A)$ denote
the set of all $r$-minors of the matrix $A$ and  $d_r(A)$ denote the g.c.d. of the polynomials of $\Delta_r(A)$. Define the ideals $\mathcal{I}_r(A)=\langle \Delta_r(A)\rangle$ and $\mathcal{J}_r(A)=\langle \Delta_r(A)/d_r(A)\rangle$ generated, respectively, by all $r$-minors $\Delta_r(A)$ and their quotients of division by $d_r(A)$. $\mathcal{I}_r(A)$ is known as the $(n-r)$-th Fitting ideal of the module $\mathrm{coker} A$. It is clear that $\mathcal{I}_r(A)\subset \mathcal{I}_{r-1}(A)$. By applying the generalized Cauchy--Binet formula one can prove the following  

\begin{proposition}
    If a unimodular matrix $U \in U_n(\PP)$ satisfies Eq. (\ref{ALA}), then 
    \begin{equation}
    \mathcal{I}_r(A)=\mathcal{I}_r(A^\dagger)\,,\qquad d_r(A)=d_r(A^\dagger)\,,\qquad \mathcal{J}_r(A)=\mathcal{J}_r(A^\dagger)\,,\qquad \forall  r=1,2,\ldots, n\,.
    \end{equation}  
\end{proposition}
For $r=n$, one gets the stronger condition $\Delta_n(A^\dagger)=\Delta_n(A)$.
In the absence of effective criteria for the solvability of equation (\ref{ALA}), we confine ourselves to studying a simpler question about the existence of an equivalent  pre-Lagrangian operator.

Following the terminology of Ref. \cite{Fabia}, we say that an operator $\hat A$ is {\it co-flat} if its symbol admits a one-step resolution 
\begin{equation}\label{A1A}
    \xymatrix{0\ar[r]&
    \PP^{l}\ar[r]^-{A_1}& \PP^{m}\ar[r]^{A}&\PP^{n}\ar[r]&\cc\, A\ar[r]&0\,,
    }
\end{equation}
where the matrix $A_1$ is unimodular. As we will see in the next section, co-flat operators are plentiful in relativistic physics.
Let us start with the following simple 
\begin{proposition}
Any two matrices $A_1$ and $A'_1$ associated with one-step resolutions of a co-flat symbol $A$ are related by the equality
$A'_1=A_1U$ for some $U\in U_{l}(\PP)$. In particular,  both matrices are unimodular and have the same size. 
\end{proposition}
\begin{proof}
Let
\begin{equation}\label{A1A'}
    \xymatrix{0\ar[r]&
    \PP^{l'}\ar[r]^-{A'_1}& \PP^{m}\ar[r]^{A}&\PP^{n}\ar[r]&\cc\, A\ar[r]&0\,,
    }
\end{equation}
be another one-step resolution of the matrix $A$.   As we know the resolutions (\ref{A1A}) and (\ref{A1A'})  are homotopy equivalent to each other. Furthermore, we can satisfy Eqs. (\ref{FGh}) with  $F_1=G_1=1_m$ and $F_{0}=G_{0}=1_n$. Then there are matrices  $F_2\in \MM_{ll'}(\PP)$ and $G_2\in \MM_{l'l}(\PP)$ such that
\begin{equation}\label{AFG}
    A_1F_2= A'_1\,,\qquad A'_1G_2=A_1\,.
\end{equation}
If $l'>l$, then  $\ker F_2\neq 0$. On the other hand, $\ker A'_1=0$ and restricting the first equation in (\ref{AFG}) onto $\ker F_2$, we arrive at the contradiction. 
Similarly,  the assumption $l>l'$ contradicts the facts that $\ker G_2\neq 0$, while  $\ker A_1=0$. Hence, $l=l'$ and both $F_2$ and $G_2$ are square matrices. 
Let $B$ be a left inverse for the matrix $A_1$. Applying $B$ to both sides of the second equation in (\ref{AFG}) yields $BA'_1G_2=1_l$. This means that the square matrix $G_2$ is unimodular and we can set $U=G_2^{-1}$. 
\end{proof}

For a general co-flat operator $\hat A$, the ranks of the free modules in (\ref{A1A}) are not related to each other. 
We say that a differential operator $\hat A\in \OO_{n,m}(\FF)$ is {\it gauge-free} if $\ker A^t=0$.  It is clear that $n\geq m$ whenever the operator $\hat A$ is gauge-free. For example, the relativistic wave equations (\ref{KG}) enjoy no gauge symmetry and the corresponding operator is gauge-free.  If $\hat A$ is a pre-Lagrangian gauge-free operator, then  $\det A\neq 0$. 

\begin{proposition}\label{p24}
    Suppose  a co-flat operator $\hat A\in \OO_{n,m}(\FF)$ is gauge-free. Then the ranks of the modules in the one-step resolution (\ref{A1A}) are related by the equality $m=l+n$. 
\end{proposition}

\begin{proof}
Applying the functor $\RR\otimes_\PP -$ to the exact sequence 
\begin{equation}\label{A1A0}
    \xymatrix{0\ar[r]&
    \PP^{l}\ar[r]^-{A_1}& \PP^{m}\ar[r]^{A}&\PP^{n}
    }
\end{equation}
yields  the exact sequence of $\RR$-vector spaces 
\begin{equation}\label{RRR}
    \xymatrix{0\ar[r]&
    \RR^{l}\ar[r]^-{A_1}& \RR^{m}\ar[r]^{A}&\RR^{n}\,.
    }
\end{equation}
The requirement for the operator $\hat A$ to be gauge-free gives two more exact sequences 
\begin{equation}
 0 \longrightarrow {\PP}^n \stackrel{A^t}{\longrightarrow}{\PP}^m \qquad \Rightarrow\qquad 0 \longrightarrow \mathcal{R}^n \stackrel{A^t}{\longrightarrow}\mathcal{R}^m\,.
 \end{equation}
 Since the dualization of finite-dimensional vector spaces preserves exactness, we also have the exact sequence  
\begin{equation}\label{RR}
 \mathcal{R}^m\stackrel{A}{\longrightarrow}\mathcal{R}^n \longrightarrow 0\,.
 \end{equation}
 By gluing   (\ref{RRR}) with (\ref{RR}), we obtain the short exact sequence 
 \begin{equation}\label{A1AR}
    \xymatrix{0\ar[r]&
    \RR^{l}\ar[r]^-{A_1}& \RR^{m}\ar[r]^{A}&\RR^{n}\ar[r]&0\,.
    }
\end{equation} 
It remains to note that every short exact sequence of finite-dimensional vector spaces necessary splits,  $\RR^m\simeq\RR^l\oplus \RR^n$, so that $m=l+n$.  Another argument: the Euler characteristic of an  acyclic complex of finite-dimensional vector spaces must be zero, i.e., $\chi=l-m+n=0$. 

 \end{proof}

\begin{proposition}\label{p25}
    Every co-flat matrix $A$ is equivalent to a matrix $B$ with $\ker B=0$. 
\end{proposition}

\begin{proof} We will present two equivalent matrices $B$ of different size. 
The first one comes from the diagram 
$$
    \xymatrixcolsep{5pc}\xymatrix{
    0\ar[r]& \PP^{l}\ar@/^/[r]^{A_1}& \PP^m\ar@/^/[r]^{A}\ar@/^/[l]^{h^A_1}\ar@/^/[d]^{G_1}& \PP^n\ar[r]^\pi\ar@/^/[l]^{h^A_0}\ar@/^/[d]^{G_{0}}& \cc\, A\\
    & 0\ar[r]& \PP^m\ar@/^/[r]^{B}\ar@/^/[u]^{F_1}& \PP^{l+n}\ar[r]^\pi\ar@/^/[l]^{h^B_{0}}\ar@/^/[u]^{F_{0}}& {\cc\, B}   }    
$$
whose homomorphisms are defined by the following matrices:
\begin{equation}
\begin{array}{llll}
    G_{0}=\binom{0}{\;1_n}\,,&F_{0}=(0, 1_n)\,,&h_0^A=0\,,&h_1^A=X\,,\\[3mm]
    G_1=1_m-A_1X\,,& F_1=1_m\,,&h_0^B=(A_1,0)\,,&B=\binom{X}{A}\,.\\
 \end{array}
\end{equation}
Here $X$ is a left inverse to $A_1$. One can easily check that the matrices satisfy Eqs. (\ref{FGh}), which   means $A\sim B$.

To construct an  alternative matrix $B$ we use the Quillen--Suslin theorem.  Let $A_1^\vee$ be a unimodular completion of $A_1$ such that the square matrix $U=(A_1, A_1^\vee)$ is invertible. Obviously, $\cc\, A=\cc\, AU$ and $A\sim AU$. On the other hand,  $AU=(0,AA_1^\vee)$ and $\mathrm{Im}\, AU=\mathrm{Im}\,AA_1^\vee$, where $AA_1^\vee: \PP^{m-l}\rightarrow \PP^n$. Therefore, $\cc\, AU=\cc\, AA_1^\vee$ and we conclude that  $A\sim AA_1^\vee$.  It remains to note that $\ker AA_1^\vee=0$. Otherwise, there would exist  a nonzero $X\in \MM_{m-l,1}(\PP)$ and $Y\in \MM_{l,1}(\PP)$ such that $A_1^\vee X=A_1Y$. But this contradicts the invertibility of $U=(A_1, A_1^\vee)$.
\end{proof}

The proposition above is a particular case of a more general statement. 

\begin{proposition}\label{p26}
    The projective dimension of the module $\cc\, A$ is equal to $N$ iff there exists a finite free resolution (\ref{ResA}) of length $N$ where the leftmost matrix $A_N$ is not unimodular.    
    \end{proposition}
For the proof, see \cite[Prop. 5.11]{book:Lam2}. In Refs. \cite{GAGOVARGAS} and \cite[Prop. 2]{Fabia},  an explicit construction was proposed to reduce the free resolution (\ref{ResA}) by one step in the case when $A_N$ is unimodular. All this makes the calculation of projective dimension a completely algorithmic  procedure. 

The projective dimension, being an  invariant of a module, provides a useful necessary condition  for a given operator $\hat{A}$ to be equivalent to a (pre-)Lagrangian one. 
\begin{proposition}\label{p27}
   If an operator $\hat A$ is equivalent to a gauge-free pre-Lagrangian operator,  then  the projective dimension of the module $\cc\, A$ is either zero or one. 
\end{proposition}
\begin{proof}
 Every gauge-free pre-Lagrangian operator  $\hat{B}\in \OO_{n,n}(\FF)$ is determined by a nondegenerate square matrix $B\in\MM_{n,n}(\PP)$. This gives  the short exact sequence $0\rightarrow \PP^n\stackrel{B}{\rightarrow}\PP^n\rightarrow\cc\,B\rightarrow 0$, meaning that the projective dimension of $\cc\, B$ is less than or equal to one. The zero projective dimension implies that the matrices $B$ and $A$ are unimodular and the differential equation $\hat Af=0$ has only the trivial solution $f=0$.  If $A\sim B$, the corresponding cokernels should have the same projective dimension. 
\end{proof}

Combining Propositions \ref{p24} and \ref{p25}, we conclude that every co-flat gauge-free matrix $A$ is equivalent to a square nondegenerate matrix $B$. In the course of the proof of Prop. \ref{p25}, two such matrices $B$ were actually constructed: one given by $\binom{X}{A}$, where $X$ is the left inverse to $A_1$, and the other given by the product $AA_1^\vee$. The former has size $m\times m$, while the latter is of dimension $n\times n$. Since $m>n$, we will refer  to the corresponding differential operators as the {\it big} and {\it small} pre-Lagrangian forms of the operator $\hat A$.

\section{Free massive fields of integer spin}
In this section,  we exemplify the above constructions by relativistic wave equations for free massive particles. As a by-product, we  get new free Lagrangians with fewer auxiliary fields. 

\subsection{Pre-Lagrangian equations for free massive fields}
Our starting point is the system of relativistic wave equations of the form 
\begin{equation}
    \label{rwe}
        T^1_{\mu_1\cdots \mu_s}:=(\square - m^2)\phi_{\mu_1\ldots\mu_s}(x) = 0\,,\qquad
         T^2_{\mu_1\cdots \mu_{s-1}}:=\partial^{\mu_s}\phi_{\mu_1\ldots\mu_{s}}(x) = 0\,.
\end{equation}
Here, $\phi_{\mu_1\ldots\mu_s}(x)$ is a fully symmetric and traceless tensor field in $d$-dimensional Minkowski space with Cartesian coordinates $x^\mu$ and all indices are raised and lowered with the (mostly plus) Minkowski metric.  Equations (\ref{rwe}) are known to be gauge-free and  describe the dynamics of free quantum particles of mass $m$. The rank of the tensor $s$ defines the spin of a particle\footnote{Of course, the relativistic wave equations (\ref{rwe}) do not cover all free massive particles for $d>4$.}. Scalar particles correspond to $s=0$, in which case the second equation in (\ref{rwe}) is absent. For $s>0$, the system is non-Lagrangian since the equations outnumber the fields.  The symbol of the operator determining Eqs. (\ref{rwe})  reads
\begin{equation}\label{symbA}
    A=\big(A^{\nu_1\cdots \nu_s}_{\mu_1\cdots\mu_s}, \, A^{\nu_1\cdots \nu_{s}}_{\mu_1\cdots \mu_{s-1}}\big)=\big( (p^2-m^2)\delta_{\mu_1}^{\nu_1}\cdots \delta_{\mu_s}^{\nu_s}, \,\delta^{\nu_1}_{\mu_1}\cdots \delta_{\mu_{s-1}}^{\nu_{s-1}}p^{\nu_s}\big)_{st}\,.
\end{equation}
The subscript $st$ stands for the symmetric and traceless projection over the upper and lower indices of the component tensors. 
The equations satisfy the obvious identity
\begin{equation}\label{id}
   \partial^{\mu_s}T^1_{\mu_1\cdots \mu_{s-1}\mu_s}-(\square -m^2)T^2_{\mu_1\cdots \mu_{s-1}}\equiv 0\,,
\end{equation}
from which one can read off the symbol $A_1$ of the compatibility operator for (\ref{rwe}). Explicitly, 
\begin{equation}\label{A1}
    A_1=\binom{\delta_{\lambda_1}^{ \mu_1}\cdots \delta_{\lambda_{s-1}}^{\mu_{s-1}}p^{\mu_s }}
    {-\delta^{\mu_1}_{\lambda_1}\cdots\delta^{\mu_{s-1}}_{\lambda_{s-1}}(p^2-m^2)}_{st}\,.
\end{equation}
It turns out that the matrix $A_1$ is unimodular or, what is the same, the matrix $A$ is co-flat. As explained in the previous section, the last fact implies that the PDE system  (\ref{rwe}) is equivalent to a pre-Lagrangian one. To construct the {\it big} pre-Lagrangian form one needs a left inverse matrix to the matrix $A_1$. In view of the tensor structure of the equations of motion (\ref{rwe}) and the identities (\ref{id}), the left inverse matrix has the block form $X=({X}_{\mu_1\cdots \mu_s}^{\nu_1\cdots \nu_{s-1}},{X}_{\mu_1\cdots \mu_{s-1}}^{\nu_1\cdots \nu_{s-1}})$, the blocks being  fully symmetric and traceless tensors in upper and lower indices. Once the matrix $X$ is identified, one can write the pre-Lagrangian 
equations as
\begin{equation}\label{erwe}
     {T}^1_{\mu_1\cdots \mu_s}+\hat{X}_{\mu_1\cdots \mu_s}^{\nu_1\cdots \nu_{s-1}}\psi_{\nu_1\cdots \nu_{s-1}}=0\,,  \qquad
         T^2_{\mu_1\cdots \mu_{s-1}}+\hat X_{\mu_1\cdots \mu_{s-1}}^{\nu_1\cdots \nu_{s-1}}\psi_{\nu_1\cdots \nu_{s-1}}= 0\,,
         \end{equation}
where $\psi_{\nu_1\cdots \nu_{s-1}}$ is a symmetric and traceless  tensor field of rank $s-1$. By construction, the total number of the initial fields $\phi$'s and the auxiliary fields $\psi$'s  is equal to the number of equations (\ref{erwe}). On substituting the extended equations (\ref{erwe}) in the identity (\ref{id}), we get the differential consequence 
\begin{equation}\label{dc}
\big(\partial^{\mu_s}\hat{X}_{\mu_1\cdots \mu_s}^{\nu_1\cdots \nu_{s-1}} - (\square -m^2)\hat X_{\mu_1\cdots \mu_{s-1}}^{\nu_1\cdots \nu_{s-1}}\big)\psi_{\nu_1\cdots \nu_{s-1}}=0\,,
\end{equation}
which involves the auxiliary fields alone. If the matrix $X$ has been chosen to satisfy $XA_1=1$, then the differential operator in (\ref{dc}) defines the identity map, which means that $\psi_{\nu_1\cdots \nu_{s-1}}=0$ and the extended equations (\ref{erwe}) are fully equivalent to the original ones (\ref{rwe}).  It should be noted that the left inverse matrix $X$ is not unique. In the case under consideration, it is natural to require the operator $\hat X$ to be Lorentz invariant. The most general Lorentz-invariant Ansatz for $\hat X$ reads
\begin{equation}\label{XX}
    \begin{array}{rcl}
\hat{X}_{\mu_1\cdots \mu_s}^{\nu_1\cdots \nu_{s-1}}\psi_{\nu_1\cdots \nu_{s-1}}&=& A_1\partial_{\{\mu_1}\psi_{\mu_2\cdots\mu_s\}} + A_2\partial_{\{\mu_1\mu_2}(\partial\psi)_{\mu_3\cdots\mu_s\}} + \ldots + A_s\partial_{\{\mu_1\cdots\mu_s\}}\partial^{s-1}\psi\,,\\[3mm]
\hat X_{\mu_1\cdots \mu_{s-1}}^{\nu_1\cdots \nu_{s-1}}\psi_{\nu_1\cdots \nu_{s-1}}&=& B_1\psi_{\mu_1\cdots\mu_{s-1}} + B_2\partial_{\{\mu_1}(\partial\psi)_{\mu_2\ldots\mu_{s-1}\}} + \ldots + B_s\partial_{\{\mu_1\cdots\mu_{s-1}\}}\partial^{s-1}\psi \,.
    \end{array}
\end{equation}
Here we introduced the following shorthand notation for repeated gradients and divergences: 
\begin{equation}\label{gdiv}
\partial_{\mu_1\cdots \mu_n}(\partial^l\Phi)_{\nu_{l+1}\cdots \nu_m}:=\partial_{\mu_1}\cdots\partial_{\mu_n}\partial^{\nu_1}\cdots \partial^{\nu_l}\Phi_{\nu_1\cdots \nu_m}\,.
\end{equation}
The braces stand for symmetrization and the traceless part of the enclosed indices, see (\ref{br}). Finally, the coefficients $A$'s and $B$'s are supposed to be polynomials in the D'Alembert operator $\square$. Substituting the Ansatz (\ref{XX}) in Eq. (\ref{dc}) and requiring the wave operator to be unity, we get a set of equations for determining the $A$'s and $B$'s. The details of the computation can be found in Appendix \ref{AppA}. The final answer is 
\begin{equation}\label{AB}
    \begin{array}{ll}
  \displaystyle   A_1=\frac1{m^2}, &\displaystyle B_1 = \frac{1}{m^2}, \\ [5mm]
  \displaystyle A_n = \frac{1}{m^2}\left(-\frac{1}{m^2}\frac{2s+d-6}{2s+d-5}\right)^{n-1}, &  \displaystyle B_n = \frac{1}{m^2}\frac{2s+d-5}{2s+d-4}\left(-\frac{1}{m^2}\frac{2s+d-6}{2s+d-5}\right)^{n-1},\\ [5mm]
  \displaystyle  A_s = -\frac{1}{m^4}\frac{d-2}{s+d-3}\left(-\frac{1}{m^2}\frac{2s+d-6}{2s+d-5}\right)^{s-2}, &  \displaystyle B_s = -\frac{1}{m^4}\frac{d-2}{2s+d-4}\left(-\frac{1}{m^2}\frac{2s+d-6}{2s+d-5}\right)^{s-2},
    \end{array}
\end{equation}
where $1<n<s$. As is seen, all  coefficients can be chosen constant. Noteworthy also are the two special cases $(s=1)$ and $(s=2, d=2)$, where only the leading coefficients $A_1=B_1=1/m^2$ are different from zero. In these cases, one can exclude the auxiliary field $\psi$ from the extended equations (\ref{erwe}) to get Lagrangian equations in terms of the original field $\phi$:
\begin{equation}\label{proca}
(\square -m^2)\phi_\mu-\partial_\mu\partial^\nu\phi_\nu=0\,,\qquad (\square -m^2)\phi_{\mu\nu}-\partial_\mu \partial^\lambda\phi_{\lambda\nu}-\partial_\nu \partial^\lambda\phi_{\lambda\mu}+\eta_{\mu\nu}\partial^\lambda\partial^\gamma\phi_{\lambda\gamma}=0\,.
\end{equation}
The first equation is known as the Proca equation for the massive vector field. As for the second equation, passing to the light-cone coordinates one can easily see that the wave operator is unimodular and the equation has only zero solution. 

One can regard the relativistic wave equations (\ref{proca}) as the {\it small} pre-Lagrangian form of the corresponding equations (\ref{rwe}). What about the small pre-Lagrangian form for the other values of $s$ and $d$? Although the existence of such a form was proved in Proposition \ref{p25}, its Lorentz invariance is not ensured. Counterintuitive as the last statement might seem\footnote{It implies that one can describe unitary irreducible representations of the Poincar\'e group by means of PDEs that are not Lorentz invariant!}, it is supported by specific examples. Consider, for instance, the case $s=2$, $d=3$. The small pre-Lagrangian form is constructed by means of the unimodular completion of the matrix $A_1$ defined by (\ref{A1}). In the case at hand, $A_1$ is an $8\times 3$ matrix, which unimodular completion  $A_1^{\vee}$ must be a matrix of size $8\times 5$. The most general Lorentz invariant Ansatz for the matrix $A_1^\vee$ is given by
\begin{equation}
    A_1^\vee=\binom{(A_1^\vee)^{\nu_1\nu_2}_{\mu_1\mu_2}}{ (A_1^\vee)^{\nu}_{\mu_1\mu_2} }=\binom{a_1\delta^{\nu_1}_{\mu_1}\delta^{\nu_2}_{\mu_2}+2a_2\delta^{\nu_1}_{\mu_1}p^{\nu_2}p_{\mu_2}+a_3p^{\nu_1}p^{\nu_2}p_{\mu_1}p_{\mu_2}}{2a_4 \delta^\nu_{\mu_1}p_{\mu_2}+a_5p_{\mu_1}p_{\mu_2}p^\nu}_{st}\,,
\end{equation}
where the coefficients $a_1,\ldots, a_5$ are polynomials in $p^2$. Computing the determinant of the $8\times 8$ matrix $(A_1, A^\vee_1)$ with \textsc{Wolfram Mathematica}, we find
\begin{equation}\label{dett}
    \det (A_1, A^\vee_1)=\frac13a_1^2\Delta^2_1(p^2)\Delta_2(p^2)\,,
    \end{equation}
    where
    \begin{equation}
    \begin{array}{l}
        \Delta_1(t)= a_1(t-m^2) + t(a_2(t-m^2) + a_4)\,,\\[3mm] 
        \Delta_2(t)=3a_1 (t-m^2) + 2t(2a_2(t-m^2) + t(a_3(t - m^2) + a_5) + 2a_4)  \,.
        \end{array}  
        \end{equation}
Each factor in (\ref{dett}) must be a nonzero constant. Evaluating the derivatives of $\Delta$'s at zero, we get
\begin{equation}
    \Delta_1'(0)=a_4-m^2a_2+a_1=0\,,\qquad \Delta'_2(0)=4(a_4-m^2a_2)+3a_1=0\,.
    \end{equation}
Therefore
\begin{equation}
    4\Delta_1'(0)-\Delta_2'(0)=a_1\neq 0
\end{equation}
and we arrive at a contradiction.  

Despite this somewhat disappointing result, the Quillen--Suslin theorem still guarantees the existence of a 
(noninvariant) unimodular completion  for the matrix $A_1$ above. Applying the package \textsc{QuillenSuslin} \cite{Fab} in the Computer Algebra System MAPLE yields the answer for $A_1^\vee$ that takes about 30 pages and involves polynomials of degree $20$. It should be noted, however, that a unimodular completion is not unique and the output largely depends on the particular algorithm implemented in the package. Whether there is a more compact expression for $A_1^\vee$ is not yet clear.     

\begin{remark}
    That equations without explicit Lorentz invariance can nevertheless describe ir\-reducible representations of  relativistic symmetry groups is nothing new. The self-dual vector field in $4$-dimensional Euclidean space provides the simplest example. The self-duality equations read
\begin{equation}\label{sd}
   T_{\mu\nu}:= F_{\mu\nu}-\frac12\varepsilon_{\mu\nu\alpha\beta}F^{\alpha\beta}=0\,,
\end{equation}
where 
    $
        F_{\mu\nu}=\partial_\mu\phi_\nu-\partial_\nu\phi_\mu 
    $
is the gauge invariant strength tensor of the vector-potential $\phi_\mu(x)$ and  $\varepsilon_{\mu\nu\alpha\beta}$ is the Levi-Civita symbol.  All indices are raised and lowered with the Euclidean metric $\delta_{\mu\nu}$. The field equations (\ref{sd}) are explicitly invariant under the action of the isometries of $ISO(4)$ and constitute an anti-self-dual tensor.  Anti-self-duality implies that only $3$ out of $6$ components of the anti-symmetric tensor $T_{\mu\nu}$ are algebraically  independent. Unfortunately, there is no $SO(4)$ invariant way to extract these three independent equations. A possible way to get rid of some redundant equations is as follows. Let us  fix a constant  $4$-vector $\xi^\mu$ of unit length, $\xi^2=1$.  The contraction of  $\xi^\mu$ with (\ref{sd}) gives the system of four equations 
\begin{equation}\label{sdx}
    T^\xi_\mu:=\xi^\nu T_{\nu\mu}=0\,.
\end{equation}
Of these equations, only three are linearly independent because of the identity $\xi^\mu T^\xi_{\mu}\equiv 0$.  
Another easily verified identity,
\begin{equation}
T_{\mu\nu}=\xi_\mu T^\xi_\nu-\xi_\nu T^\xi_\mu-\frac12\varepsilon_{\mu\nu}{}^{\alpha\beta} (\xi_\alpha T^\xi_\beta-\xi_\beta T^\xi_\alpha)\,,
\end{equation}
 shows that Eqs. (\ref{sd}) and (\ref{sdx}) are completely equivalent to each other and have the same solutions.
In other words, the solution space of equation (\ref{sdx}) is invariant under the action of $ISO(4)$, even though the presence of a given $4$-vector $\xi$ violates the explicit $SO(4)$ invariance. 
 \end{remark}

We conclude this subsection with a brief comment on the massless limit of equations (\ref{rwe}). Although the value $m=0$ leads to an unsatisfactory physical theory, it can be used to illustrate some mathematical concepts of the previous section. In that case, the matrix of the compatibility operator (\ref{A1}), being of full rank, is not unimodular any more. The reason is simple:  it vanishes at $p=0$. Therefore, the resulting system is not co-flat. The corresponding symbol 
matrices $A$ and $A_1$ still define a free resolution of the module $\mathrm{coker}\,A$. By Proposition \ref{p26} this module has projective dimension two.  
Then it follows from Proposition \ref{p27} that any Lagrangian for massless equations (\ref{rwe}), if it exists, must necessarily involve gauge fields. 

\subsection{New Lagrangians for free massive fields}
Having brought the relativistic wave equations (\ref{rwe}) into  a pre-Lagrangian form, one may wonder about their  Lagrangian formulation. Notice that the solution (\ref{XX}, \ref{AB}) does not give a Lagrangian system of equation, since the operator of the extended system (\ref{erwe}) is not formally self-adjoint. It should be 
remembered that the matrix $X$ -- a left inverse to $A_1$ -- is not unique.  Furthermore, any relativistic wave operator is defined up to a unimodular factor.  Given this ambiguity, we start with the  most general Lorentz invariant Ansatz for the Lagrangian:
\begin{equation}\label{L}
    \mathcal{L} = \sum_{n=0}^s(-1)^n (\partial^n\phi) \cdot A_n(\partial^n\phi) + \sum_{n=0}^{s-1}(-1)^n (\partial^n\psi) \cdot B_n(\partial^n\psi)+\sum_{n=1}^s(-1)^n(\partial^n\phi)\cdot C_n(\partial^{n-1}\psi)\,.    
    \end{equation}
Here the dots stand for the full contraction of tensor indices with the help of the Minkowski metric and the coefficients $A_n$, $B_n$, and $C_n$ are polynomials in $\square$ to be determined.
    The corresponding Euler--Lagrange equations read 
\begin{equation}\label{ELE}
    \begin{array}{l}
       \displaystyle A_0\phi_{\mu_1\ldots\mu_s} + \sum\limits_{n = 1}^s\Big(A_n\partial_{\{\mu_1\cdots\mu_n}(\partial^n\phi)_{\mu_{n+1}\cdots\mu_s\}} + \frac{1}{2}C_n\partial_{\{\mu_1\cdots\mu_n}(\partial^{n-1}\psi)_{\mu_{n+1}\cdots\mu_s\}}\Big) = 0\,,\\[5mm]
      \displaystyle  B_0\psi_{\mu_1\ldots\mu_{s-1}} + \sum\limits_{n = 1}^{s-1}B_n\partial_{\{\mu_1\cdots\mu_n}(\partial^n\psi)_{\mu_{n+1}\ldots\mu_{s-1}\}}
      - \sum\limits_{n = 1}^{s}\frac{1}{2}C_n\partial_{\{\mu_1\cdots\mu_{n-1}}(\partial^{n}\phi)_{\mu_{n}\cdots\mu_{s-1}\}} = 0\,.
    \end{array}
\end{equation}
Our strategy  is to fix the free coefficients by requiring the field equations 
\begin{equation}
(\square-m^2)\phi_{\mu_1\cdots\mu_s}=0\,,\qquad \partial^{\mu_1}\phi_{\mu_1\cdots\mu_s}=0\,,\qquad \psi_{\mu_1\cdots\mu_{s-1}}=0\,.
\end{equation}
to be among the differential consequences of (\ref{ELE}). As explained in  Appendix \ref{AppB}, it is enough to consider only successive divergences of equations (\ref{ELE}). This leads to a system of quadratic equations on the coefficients. Below we present some explicit solutions for lower spins. In all the cases $A_0 = \square - m^2$ and the remaining coefficients are constants. 

\paragraph{ Spin $2$:}
\begin{equation}\label{s2}
    \begin{array}{lll}
     \displaystyle
         A_1 =-\frac{d}{2(d-1)}\,,& \displaystyle B_0 = -\frac{(d-1)}{2(d-2)}C_1^2\,,\quad& \\[5mm]
 \displaystyle A_2 = 0\,,\quad& \displaystyle B_1 = \frac{d-1}{2  m^2 d}C_1^2\,,&C_2 = 0\,.\\
    \end{array}
\end{equation}

\paragraph{ Spin $3$:}

\begin{equation}\label{s3-1}
    \begin{array}{lll}
 \displaystyle A_1 = \frac{d^4-6 d^2-2 d+4}{2 d^2+2 d}\,,& \displaystyle B_0 = -\frac{m^4(d+1)}{2d(d-2)^2}C_2^2\,,&         \\[5mm]
 \displaystyle A_2 = -\frac{(d-2)^2 (d+2)}{2 m^2 (d+1) }\,,\quad& \displaystyle B_1 = \frac{m^2(d+1)}{2(d+2)(d-2)^2}C_2^2\,,\quad& 
 \displaystyle C_1 = -\frac{m^2(d^2-2)}{d(d-2)}C_2\,,\\[5mm]
 \displaystyle A_3 = \frac{(d-2)^3 (d+2)}{2 m^4(d+1)^2 }\,,& \displaystyle B_2 = -\frac{1}{2(d+2)(d-2)}C_2^2\,,& 
 \displaystyle C_3 = -\frac{(d-2)d^2}{m^2d^2(d+1)}C_2\,,
    \end{array}
    \end{equation}
   or
   \begin{equation}\label{s3-2}
    \begin{array}{lll}
 \displaystyle A_1 = -\frac{11 d^4+12 d^3+8 d^2+24 d-64}{12 d (d+1) \left(d^2+2\right)}\,,& \displaystyle B_0 = -\frac{3m^4(d+1)(d^2+2)}{4d(d-2)^2}C_2^2\,,&     \\[5mm]
 \displaystyle A_2 = -\frac{(d-2)^2 (d+2)}{3 m^2 (d+1) (d^2+2) }\,,& \displaystyle B_1 = \frac{3m^2(d+1)(d^2+2)}{4(d-2)^2(d+2)}C_2^2\,,& \displaystyle C_1 = \frac{m^2(d^2+8) }{2 (d-2) d}C_2\,,\\[5mm]
 \displaystyle A_3 = \frac{3 (d-2)^3 (d+2)}{4 m^4 (d+1)^2 (d^2+2) }\,,& 
 \displaystyle B_2 = -\frac{3(d^2+2)}{4(d-2)(d+2)}C_2^2\,,& \displaystyle C_3 = -\frac{3 (d-2)}{2 m^2(d+1) }C_2\,.
    \end{array}
\end{equation}
\paragraph{ Spin $4$:}
\begin{equation}\label{s4}
    \begin{array}{ll}
 \displaystyle A_1 = \frac{d^4-40 d^2-104 d-32}{4 (d^3+9 d^2+22 d+12)}\,, &
 \displaystyle B_0 = -\frac{m^4 (d^3+9 d^2+22 d+12) }{4 d^2 (d+2)^2}C_2^2\,,
         \\[5mm]
 \displaystyle A_2 = -\frac{d^2 (d^2+6 d+8)}{m^2(d^3+9 d^2+22 d+12) }\,, & 
 \displaystyle B_1 = \frac{m^2(d^3+9 d^2+22 d+12)}{4 d^2 (d+2) (d+4)}C_2^2\,, \\[5mm]
         A_3 = A_4 = 0\,,&  B_2 = B_3 = 0\,, \\[5mm]
\displaystyle C_1 = \frac{m^2(d^2+2 d-4) }{2 d (d+2)}C_2\,, & C_3 = C_4 = 0\,.
    \end{array}
\end{equation}
The coefficients $C_1$ in (\ref{s2}) and $C_2$ in (\ref{s3-1} -- \ref{s4}) are arbitrary nonzero constants. It should be noted that the solutions above are by no means unique. For instance, one can change the coefficients by merely rescaling the fields: $\phi\rightarrow \alpha\phi$, $\psi\rightarrow \beta\psi$. 
In \cite{PhysRevD.18.4548}, O'Brien  considered the special case $s=3,d=4$.  Requiring that the Euler--Lagrange equations be of order $4$, he found another solution wherein some of the coefficients, besides $A_0$, are nonconstant polynomials in $\square$. Numerical experiments using \textsc{Wolfram Mathematica} indicate that beyond spin $4$ there are no solutions in which all coefficients except $A_0$ are constant.

\subsection{Comparison with the Fierz--Pauli theory}

Let us consider in more detail the case $s=2$ and $d=4$, studied originally by Fierz and Pauli \cite{Fierz:1939ix}.  The difference between our model
and that of Ref. \cite{Fierz:1939ix} is in the auxiliary fields: instead of a single scalar field $\chi$ used by Fierz and Pauli, we introduce the vector field $\psi_\mu$. In order to systematically relate both models, we augment system (\ref{rwe}) for $s=2$ with the divergence of the second equation. This gives the equivalent system of equations
\begin{equation}
    \label{TTT}
        T_{\mu \nu}:=(\square - m^2)\phi_{\mu \nu} = 0\,, \qquad  T_{\nu}:= \partial^\mu\phi_{\mu\nu} = 0\,, \qquad T:=\partial^\mu\partial^\nu\phi_{\mu\nu}= 0\,.
\end{equation}
The augmented system enjoys the obvious identities:  
\begin{equation}\label{TT}
        \partial^{\mu}T_{\mu \nu} - (\square - m^2)T_{\nu}\equiv 0\,, \qquad \partial^{\nu}T_{\nu} - T\equiv 0\,.
\end{equation}
It is easy to see that the augmented system (\ref{TTT}) is co-flat. By Proposition \ref{p25} it admits a pre-Lagrangian form.  The tensor structure of the identities (\ref{TT}) suggests that the big pre-Lagrangian form of (\ref{TTT}) involves one vector field $\psi_{\nu}$ and one scalar field $\chi$. One can also think of $\chi$ as a generalized Stueckelberg field in the sense of Ref. \cite{abakumova2023dualisation}. 
We make the following choice for the extended system: 
\begin{equation}\label{FPee}
   \begin{array}{l}
        (\square - m^2)\phi_{\mu\nu} + A(\partial_{\mu}\psi_{\nu} + \partial_{\nu}\psi_{\mu} - \frac{1}{2}\eta_{\mu\nu}\partial^\lambda\psi_\lambda) \\[3mm]
        + B(\partial_{\mu}\partial_{\nu}\partial\psi - \frac{1}{4}\eta_{\mu\nu}\square\partial^\lambda\psi_\lambda) + C(\partial_{\mu}\partial_{\nu}\chi - \frac{1}{4}\eta_{\mu\nu}\square\chi) = 0\,, \\[3mm]
        \partial^\mu\phi_{\mu\nu} + D \psi_{\nu} + E\partial_{\nu}\partial^\mu\psi_\mu + F \partial_{\nu}\chi = 0\,, \\[3mm]
        \partial^{\mu}\partial^{\nu}\phi_{\mu\nu} + G\partial^\mu\psi_\mu + H\chi = 0\,.
    \end{array}
\end{equation}
Again, the scalar coefficients $A,B,\ldots, H$ are some polynomials in $\square$ to be determined from the requirement that the extended system be equivalent to the original one. On substituting the extended equations (\ref{FPee}) into the identities (\ref{TT}), we obtain the system for the auxiliary fields alone:
\begin{equation}\label{eq2}
    \begin{array}{l}
        \big((A- D)\square + m^2 D\big)\psi_{\nu} + \big(\frac{1}{2}A + \frac{3}{4}B\square - E(\square - m^2) \big)\partial_{\nu}\partial^\mu\psi_\mu + \big(\frac{3}{4}C\square - F\square +m^2 F\big)\partial_{\nu}\chi = 0\,, \\[3mm]
        (D + E\square - G)\partial^\mu\psi_\mu + (F\square - H)\chi = 0\,.
    \end{array}
\end{equation}
Combining the second equation with the divergence of the first gives the pair of scalar equations
\begin{equation}\label{s2sc}
    \begin{array}{l}
        \Big(m^2D + \big(\frac{3}{2}A - D + m^2E\big)\square + \big(\frac{3}{4}B - E\big)\square^2\Big)\partial^\mu\psi_\mu + \big(\frac{3}{4}C\square - F\square +m^2 F\big)\square\chi = 0\,,\\[3mm]
        (D + E\square - G)\partial^\mu\psi_\mu + (F\square - H)\chi = 0\,.
    \end{array}
\end{equation}
At this point, we make a simplifying assumption that $E=0$ and $H = H_1 + H_2 \square$  where $H_1$, $H_2$ and all other coefficients  $A,\ldots, G$ are constant. If we now want the linear  equations (\ref{s2sc}) to have only a trivial solution, $\partial^\mu\psi_\mu=\chi=0$, the operator determinant of the system 
\begin{equation}
    \begin{array}{l}
         \mathrm{Det} = -m^2DH_1 +\left(DH_1 - m^2DH_2 + m^2FG - \frac{3}{2}AH_1\right)\square\\[3mm]
         +\left(\frac{3}{2}AF + \frac{3}{4}CG - \frac{3}{4}CD - FG - \frac{3}{4}BH_1 + DH_2 - \frac{3}{2}AH_2\right)\square^2 + \frac{3}{4}B\left(F - H_2\right)\square^3 \,,
    \end{array}
\end{equation}
must be a nonzero constant. Then the vector equation in (\ref{eq2}) simplifies to 
\begin{equation}
    \left((A- D)\square + m^2 D\right)\psi_{\nu} = 0
\end{equation}
It gives the desired equation $\psi_\nu=0$ provided $A = D\neq 0$. This leads to the following system of equations:
\begin{equation}
    \begin{array}{l}
        A = D\,,\quad 
        F = H_2\,, \quad       DH_1 - m^2DH_2 + m^2FG - \frac{3}{2}AH_1 = 0\,,\\[3mm]
        \frac{3}{2}AF + \frac{3}{4}CG - \frac{3}{4}CD - FG - \frac{3}{4}BH_1 + DH_2 - \frac{3}{2}AH_2 = 0\,.
    \end{array}
\end{equation}
We can satisfy this system by setting
\begin{equation}
        A = D = \frac{m}{2}\,,\quad 
        B = -\frac{1}{3m}(C + 1)\,,\quad
        F = H_2 = -\frac{3}{4}\,,\quad
        G = m\,,\quad
        H_1 = -\frac{3m^2}{2}\,,
\end{equation}
with $C$ being an arbitrary constant. Then the field equations (\ref{FPee}) take the form
\begin{equation}
    \begin{array}{l}
        (\square - m^2)\phi_{\mu\nu} + \frac{m}{2}(\partial_{\mu}\psi_{\nu} + \partial_{\nu}\psi_{\mu} - \frac{1}{2}\eta_{\mu\nu}\partial^\lambda\psi_\lambda) \\[3mm]
        - \frac{1}{3m}(C + 1)(\partial_{\mu}\partial_{\nu}\partial\psi - \frac{1}{4}\eta_{\mu\nu}\square\partial^\lambda\psi_\lambda) + C(\partial_{\mu}\partial_{\nu}\chi - \frac{1}{4}\eta_{\mu\nu}\square\chi) = 0\,, \\[3mm]
        \partial^\mu\phi_{\mu\nu} + \frac{m}{2}\psi_{\nu} -\frac{3}{4}\partial_{\nu}\chi = 0\,, \\[3mm]
        \partial^{\mu}\partial^{\nu}\phi_{\mu\nu} + m\partial^\mu\psi_\mu - \frac{3}{4}(\square + 2m^2)\chi = 0\,.
    \end{array}
\end{equation}
Solving the second equation for $\psi_\nu$, we get
\begin{equation}
    \psi_{\nu} = \frac{3}{2m}\partial_{\nu}\chi - \frac{2}{m}\partial^\mu\phi_{\mu\nu}\,.
\end{equation}
On substituting this $\psi_\nu$ into the remaining two equations, we arrive at the PDE system 
\begin{equation}
    \begin{array}{l}
    (\square - m^2)\phi_{\mu\nu} - \big(\partial_{\mu}(\partial\phi)_{\nu} + \partial_{\nu}(\partial\phi)_{\mu} - \frac{1}{2}\eta_{\mu\nu}\partial^2\phi\big) + \frac{1}{2}\big(\partial_{\mu}\partial_{\nu}\chi - \frac{1}{4}\eta_{\mu\nu}\square\chi\big) = 0\,,\\[3mm]
         \partial^\mu\partial^\nu\phi_{\mu\nu} - \frac{3}{4}(\square - 2m^2)\chi = 0\,,
    \end{array}
\end{equation}
which is identical to the Fierz--Pauli equations for the massive spin-2 field. As is seen, the free parameter $C$ drops from the final expressions. For other Lagrangian formulations of the spin-2 field,  see \cite{PhysRev.148.1259, PhysRevD.6.984, abakumova2023dualisation, GRIGORIEV2022115686}. 

\section{Conclusion}

Let us summarize the main result of our paper in the following thesis:   {\it A system of PDEs with constant coefficients is equivalent to a 
pre-Lagrangian system without gauge symmetry if and only if the dual solution module has projective dimension one.} Of course, the absence of a pre-Lagrangian formulation implies the absence of a Lagrangian one. The relativistic wave equations (\ref{KG}) exemplify both possibilities depending on whether $m^2>0$ or $m^2=0$. 

There are several natural directions for further generalization and application of our results. 
In Sec. 3, we restricted our consideration to fully symmetric tensor fields in $d$-dimensional Minkowski space. 
For $d>4$, such fields do not  exhaust all unitary irreducible representations of the Poincar\'e group. 
Therefore, it is interesting to apply the above analysis to mixed-symmetry fields. 
This does not require any modifications to the method. Another possible direction would be an extension of the above homological analysis to 
gauge invariant systems of PDEs. Even if the original PDE system enjoys no gauge symmetry it may well be equivalent to a gauge invariant (pre-)Lagrangian system.  Specific examples include frame-like  formulations of massive fields \cite{Z1, PV, Z2}. 
To cover the case of gauge invariant systems  one needs the concept of `two-side resolutions' of wave operators, introduced and studied in \cite{KLSh1, KLSh2}. Unfortunately, such resolutions have not yet received much attention in the mathematical literature. Here, we can only mention the paper \cite{AN}, which deals 
with two-side resolutions under very restrictive assumptions, meaning no physical degrees of freedom. Finally, to consider more general backgrounds than Minkowski space, we need to generalize our homological analysis  to differential operators with non-constant coefficients. This brings us immediately to the realm of D-modules and their resolutions. Although the analysis becomes more complex, it is still doable  with modern computer algebra systems.

\section*{Acknowledgements}

The work of A. Sh. was partially supported by the São Paulo Research Foundation (FAPESP), grant 2022/13596-8. D. Sh. acknowledges the financial support from the Foundation for the Advancement of Theoretical Physics and Mathematics ``BASIS''. 
The results of Sec. 3.2 were obtained under the exclusive support of the Tomsk State University Development Program (Priority–2030).

\appendix

\section{Pre-Lagrangian form of free massive theories}\label{AppA}

We use the notation (\ref{gdiv}) for repeated gradients and divergences of the tensor $\psi_{\mu_2\cdots \mu_s}$ of rank $s-1$.  The symmetrization of covariant indices is defined by the formula 
\begin{equation}
    \partial_{(\mu_1\cdots\mu_r}(\partial^{r-1}\psi)_{\mu_{r+1}\ldots\mu_s)} := \frac{1}{r!(s-r)!}\sum\limits_{\sigma \in S_n}\partial_{\sigma(\mu_1)\cdots\sigma(\mu_r)}(\partial^{r-1}\psi)_{\sigma(\mu_{r+1})\ldots\sigma(\mu_s)}
\end{equation}
and the projection on the traceless part is denoted by braces: 
\begin{equation}\label{br}
    \partial_{\{\mu_1\cdots\mu_r}(\partial^{r-1}\psi)_{\mu_{r+1}\ldots\mu_s\}} := \partial_{(\mu_1\cdots\mu_r}(\partial^{r-1}\psi)_{\mu_{r+1}\ldots\mu_s)} - \frac{1}{2s + d - 4}\eta_{(\mu_1\mu_2}\Theta _{\mu_3\ldots\mu_s)}\,.
\end{equation}
Here the tensor $\Theta_{\mu_3\ldots\mu_s}$ is defined by induction on the number of gradients as
\begin{equation}
    \Theta_{\mu_3\ldots\mu_s} = \square\partial_{\{\mu_3\cdots\mu_r}(\partial^{r-1}\psi)_{\mu_{r+1}\ldots\mu_s\}} + 2\partial_{\{\mu_3\cdots\mu_{r+1}}(\partial^r\psi)_{\mu_{r+2}\ldots\mu_s\}}\,.
\end{equation}
The base of the induction is
\begin{equation*}
    (\partial^{r-1}\psi)_{\{\mu_{r+1}\ldots\mu_s\}} = (\partial^{r-1}\psi)_{\mu_{r+1}\ldots\mu_s} \,.
\end{equation*}
The most general Lorentz invariant Ansatz  for the extended system of equations reads 
\begin{equation}
    \begin{array}{l}
        (\square - m^2)\phi_{\mu_1\ldots\mu_s} + A_1\partial_{\{\mu_1}\psi_{\mu_2\ldots\mu_s\}}  + \ldots + A_s\partial_{\{\mu_1\mu_2\cdots\mu_s\}}(\partial^{s-1}\psi) = 0\,,\\[3mm]
        (\partial\phi)_{\mu_2\ldots\mu_s} + B_1\psi_{\mu_2\ldots\mu_s} + B_2\partial_{\{\mu_2}(\partial\psi)_{\mu_3\ldots\mu_s\}} + \ldots + B_s\partial_{\{\mu_2\ldots\mu_s\}}(\partial^{s-1}\psi) = 0\,.
    \end{array}
\end{equation}
To study the differential consequences of these equations we use the following formulas for divergences: 
\begin{equation}\label{DIV}
\begin{array}{rcl}
    \partial^{\mu_1}\partial_{\{\mu_1}\psi_{\mu_2\ldots\mu_s\}} &=&  \displaystyle \square\psi_{\mu_2\ldots\mu_s} + \frac{2s + d - 6}{2s + d - 4}\partial_{\  \mu_2}(\partial\psi)_{\mu_3\ldots\mu_s\}}\,,\\[5mm]
\partial^{\mu_1}\partial_{\{\mu_1\cdots\mu_r}(\partial^{r-1}\psi)_{\mu_{r+1}\ldots\mu_s\}} &= &  \displaystyle \frac{2s + d - 5}{2s + d - 4}\square\partial_{\{\mu_2\cdots\mu_r}(\partial^{r-1}\psi)_{\mu_{r+1}\ldots\mu_s\}}\,,\\[5mm]
    &+&  \displaystyle \frac{2s + d - 6}{2s + d - 4}\partial_{\{\mu_2\cdots\mu_{r+1}}(\partial^r\psi)_{\mu_{r+2}\ldots\mu_s\}}\,,\quad 1<r<s-1\,,\\[5mm]
    \partial^{\mu_1}\partial_{\{\mu_1\cdots\mu_{s-1}}(\partial^{s-2}\psi)_{\mu_{s\}}} &=&  \displaystyle \frac{2s + d - 5}{2s + d - 4}\square\partial_{\{\mu_2\cdots\mu_s-1}(\partial^{s-2}\psi)_{\mu_s\}} 
    \\[5mm]&+ &  \displaystyle \frac{d - 2}{2s + d - 4}\partial_{\{\mu_2\cdots\mu_s\}}(\partial^{s-1}\psi)\,,\\[5mm]
    \partial^{\mu_1}\partial_{\{\mu_1\cdots\mu_s\}}(\partial^{s-1}\psi) &=&  \displaystyle \frac{s + d - 3}{2s + d - 4}\square\partial_{\{\mu_2}\ldots\partial_{\mu_s\}}(\partial^{s-1}\psi)\,.
    \end{array}
\end{equation}
With these relations, Eq. (\ref{dc}) takes the form 
\begin{equation}
\begin{array}{c}
    \big(m^2B_1 + (A_1 - B_1)\square\big)\psi_{\mu_2\ldots\mu_s} \\[5mm]
  \displaystyle + \sum\limits_{r = 2}^{s-1}\Big(\frac{2s + d - 5}{2s + d - 4}A_r - B_r\Big)\square\partial_{\{\mu_2\cdots\mu_r}(\partial^{r-1}\psi)_{\mu_{r+1}\ldots\mu_s\}} \\[5mm]
  \displaystyle +\sum\limits_{r=1}^{s-2}\Big(\frac{2s + d - 6}{2s + d - 4}A_r + m^2B_{r+1}\Big)\partial_{\{\mu_2\cdots\mu_{r+1}}(\partial^r\psi)_{\mu_{r+2}\ldots\mu_s\}}\\[5mm]
  \displaystyle + \Big(\frac{d - 2}{2s + d - 4}A_{s-1} + m^2B_s\Big)\partial_{\{\mu_2\cdots\mu_s\}}(\partial^{s-1}\psi)\\[5mm]
  \displaystyle + \Big(\frac{s + d - 3}{2s + d - 4}A_s - B_s\Big)\square\partial_{\{\mu_2\cdots\mu_s\}}(\partial^{s-1}\psi) = 0\,.
    \end{array}
\end{equation}
This leads to the following equations on the coefficients:
\begin{equation}
    \begin{array}{ll}
m^2B_1 + (A_1 - B_1)\square=1\,,&\\[3mm]
  \displaystyle \frac{2s + d - 5}{2s + d - 4}A_r - B_r = 0\,,&1 < r < s\,,\\[3mm]
  \displaystyle \frac{2s + d - 6}{2s + d - 4}A_r + m^2B_{r+1} = 0\,, & r < s-1\,,\\[3mm]
  \displaystyle \frac{d - 2}{2s + d - 4}A_{s-1} + m^2B_s = 0\,,&
  \displaystyle \frac{s + d - 3}{2s + d - 4}A_s - B_s = 0\,,
\end{array}
\end{equation}
whence 
\begin{equation*}
    \frac{2s + d - 5}{2s + d - 4}A_{r+1} = B_{r+1} = -\frac{2s + d - 6}{m^2(2s + d - 4)}A_r\,, \quad 1<r<s\,.
\end{equation*}
The last formula gives the recurrence relations for determining the coefficients. Setting, for example,  $A_1 = {1}/{m^2}$, we arrive at the solution (\ref{AB}).

\section{Derivation of the Lagrangians}\label{AppB}

Our starting point is the Lagrangian  (\ref{L}) together with the field equations (\ref{ELE}). By making use of formulas (\ref{DIV}) we compute the successive divergences of equations (\ref{ELE}). The differential consequences with $s-n$ indices have the following structure:
\begin{equation}
    \label{gen}
    \begin{array}{l}
        A^{(n)}(\partial^n\phi)_{\mu_1\ldots\mu_{s-n}} + \frac{1}{2}C^{(n)}(\partial^{n-1}\psi)_{\mu_1\ldots\mu_{s-n}} + \ldots = 0\,,\\[3mm]
        B^{(n)}(\partial^{n-1}\psi)_{\mu_1\ldots\mu_{s-n}} - \frac{1}{2}D^{(n)}(\partial^n\phi)_{\mu_1\ldots\mu_{s-n}} + \ldots = 0\,.
    \end{array}
\end{equation}
Here
\begin{equation}
\begin{array}{ll}
      \displaystyle  A^{(n)}= \sum_{k=0}^n f_k(s,n)A_k\square^k\,,&\qquad  \displaystyle 
      B^{(n)} =\sum_{k=0}^{n-1} f_k(s, n-1)B_k\square^k\,,\\[5mm]
  \displaystyle C^{(n)} = \sum_{k=1}^n f_k(s,n) C_k\square^k\,,&\qquad
  \displaystyle D^{(n)} = \sum_{k=0}^{n-1}f_k(s, n-1)C_{k+1}\square^{k}\,,
    \end{array}
\end{equation}
and the coefficients $f_k(s,n)$ are defined by the relations 
\begin{equation*}
    f_0(s,n) = 1\,,\qquad f_1(s,n) = \frac{(2s+d-3-n)n}{2s+d-4}\,,\qquad f_s(s,s) = \prod_{i = 0}^{s-2} \frac{s-i + d -3}{2(s-i) + d - 4}\,,
\end{equation*}
\begin{equation*}
    f_k(s,k) = \prod_{i = 0}^{k-2} \frac{2(s-i)+d-5}{2(s-i)+d-4}\qquad 1 < k < s\,,
\end{equation*}
\begin{equation*}
    f_k(s,n) = \frac{2s+d-6}{2s+d-4}f_k(s-1,n-1) + \frac{2s+d-5}{2s+d-4}f_{k-1}(s-1,n-1) \qquad 1 < k < n \leq s,\quad s-k>1\,,
\end{equation*}
\begin{equation*}
    f_k(s,n) = \frac{d-3}{2s+d-4}f_k(s-1,n-1) + \frac{2s+d-5}{2s+d-4}f_{k-1}(s-1,n-1) \qquad s = n = k + 1\,.
\end{equation*}
The dots in (\ref{gen}) stand for the terms involving $\partial^{m}\phi$ and $\partial^{l}\psi$ with $m>n$ and $l>n-1$. These last terms are supposed to vanish 
due to equations with less number of free indices. Under this assumption, (\ref{gen}) becomes a homogeneous linear system of equations on the unknowns $\partial^n\phi$ and $\partial^{n-1}\psi$.  In order for this system to have only the trivial solution, the operator determinant of the system
\begin{equation}
    \mathrm{Det}_n=A^{(n)}B^{(n)} + \frac{1}{4}C^{(n)}D^{(n)}
\end{equation}
must be a nonzero constant. Evaluating the determinants for all $n=1,\ldots, s$ gives a system of quadratic equations on the coefficients $A_k$, $B_k$, and $C_k$. Unfortunately, we were unable to solve these equations explicitly for an arbitrary value of spin. For some lower spins, the equations run as follows. In all cases $A_0=t-m^2$ and
\begin{equation}
\mathrm{Det}_1=\left| 
\begin{array}{cc}
(A_1+1)t -m^2& \frac12C_1 \\
   -\frac12C_1  & B_0
\end{array}\right|\,,
\end{equation}
where $t=\square$. The other determinants depend on the value of spin and dimension. 

\paragraph{Spin 2:}
\begin{equation}
    \mathrm{Det}_2 =
    \begin{vmatrix}
        \frac{d-1}{d}(A_2t^2 + 2A_1t) + t - m^2& \frac{d-1}{2d}(C_2t^2 + 2C_1t)\\[3mm]
        -\frac{1}{2}(C_1 + C_2t)& B_0 + B_1t
    \end{vmatrix}\,.
\end{equation}
\paragraph{Spin 3:}
\begin{equation}
\begin{array}{l}
    \mathrm{Det}_2 =
    \begin{vmatrix}
        \frac{d+1}{d+2}(A_2t^2 + 2A_1t) -m^2 + t& \frac{d+1}{2(d+2)}(C_2t^2 + 2C_1t)\\[3mm]
- \frac{1}{2}(C_1 + C_2t)& B_0 + B_1t
    \end{vmatrix}\,,\\[7mm]
    \mathrm{Det}_3 = 
    \begin{vmatrix}
        \frac{(d-1)A_3t^3 + 3(d-1)A_2t^2 + 3dA_1t}{d+2} + t -m^2& \frac{(d-1)C_3t^3 + 3(d-1)C_2t^2 + 3dC_1t}{2(d+2)}\\[3mm]
-\frac{d-1}{2d}(C_3t^2 + 2C_2t) - \frac{1}{2}C_1&\frac{d-1}{d}(B_2t^2 + 2B_1t)+B_0
    \end{vmatrix}\,.
\end{array}
\end{equation}

\paragraph{Spin 4:}
\begin{equation}
\begin{array}{l}
    \mathrm{Det}_2 = 
    \begin{vmatrix}
        \frac{d+3}{d+4}(A_2t^2 + 2A_1t) -m^2 + t& \frac{d+3}{2(d+4)}(C_2t^2 + 2C_1t)\\[3mm]
 - \frac{1}{2}(C_1 + C_2t)& B_0 + B_1t
    \end{vmatrix}\\[7mm]
    \mathrm{Det}_3 = 
    \begin{vmatrix}
        \frac{\alpha_3A_3t^3 + \alpha_2A_2t^2 + \alpha_1A_1t}{(d+2)(d+4)} + t - m^2& \frac{\alpha_3C_3t^3 + \alpha_2C_2t^2 +\alpha_1C_1t}{2(d+2)(d+4)}\\[3mm]
 - \frac{1}{2}(\frac{d+1}{d+2}(C_3t^2 + 2C_2t) + C_1)&\frac{d+1}{d+2}(B_2t^2 + 2B_1t)+B_0
    \end{vmatrix}\,,\\[7mm]
    \mathrm{Det}_4 = 
    \begin{vmatrix}
        \frac{\beta_4A_4t^4 + \beta_3A_3t^3 + \beta_2A_2t^2 + \beta_1A_1t}{(d+2)(d+4)} + t - m^2& \frac{\beta_4C_4t^4 + \beta_3C_3t^3 + \beta_2C_2t^2 + \beta_1C_1t}{(d+2)(d+4)}\\[3mm]
        \frac{(d-1)B_3t^3 + 3(d-1)B_2t^2 + 3dB_1t}{d+2}+B_0& - \frac{1}{2}(\frac{(d-1)C_4t^3 + 3(d-1)C_3t^2 + 3dC_2t}{d+2} + C_1)
    \end{vmatrix}\,,
    \end{array}
\end{equation}
where
\begin{equation}
\begin{array}{l}
    \alpha_3 = d^2+4d+3\,,\quad \alpha_2 = 3d^2 + 11d + 8\,,\quad \alpha_1 = 3d^2 + 12d + 12\,,\\[3mm]
 \beta_4 = d^2 - 1\,,\quad \beta_3 = 4d^2 + 3d - 7\,,\quad \beta_2 = 6d^2 + 12d - 6\,,\quad \beta_1 = 4d^2 + 12d + 8\,.
    \end{array}
\end{equation}
In the equations above, all $\mathrm{Det}$'s are required to be nonzero constants. One can look for the coefficients $A_k$, $B_k$, and  $C_k$ that are real numbers. This leads to overdetermined systems of quadratic equations. Applying {\sc Wolfram Mathematica} yields the explicit solutions (\ref{s2}--\ref{s4}).

\end{document}